\documentclass[letterpaper,11pt]{article}
\usepackage{verbatim}
\usepackage{xspace}
\usepackage{amsmath}
\newtheorem{thm}{Theorem}
\newtheorem{cor}[thm]{Corollary}
\newtheorem{lem}[thm]{Lemma}

\def\squarebox#1{\hbox to #1{\hfill\vbox to #1{\vfill}}}
\newcommand{\qed}{\hspace*{\fill}
\vbox{\hrule\hbox{\vrule\squarebox{.667em}\vrule}\hrule}\smallskip}
\newenvironment{proof}{\noindent{\bf Proof:~~}}{\(\qed\)}

\begin{document}

\author{Erick Chastain\footnote{Computer Science Department, Rutgers University, erickc@cs.rutgers.edu}
\and Adi Livnat\footnote{Department of Biological Sciences, Virginia Tech, adi@vt.edu}
\and Christos Papadimitriou\footnote{Computer Science Division, University of California at Berkeley, CA, 94720 USA, christos@cs.berkeley.edu}
\and Umesh Vazirani\footnote{Computer Science Division, University of California at Berkeley, CA, 94720 USA, vazirani@cs.berkeley.edu}
}

\title{Multiplicative Updates in Coordination Games \\
and the Theory of Evolution}

\date{}

\maketitle

\begin{abstract}
\noindent We study the population genetics of Evolution in the important special case of weak selection, in which all fitness values are assumed to be close to one another.  We show that in this regime natural selection is tantamount to the multiplicative updates game dynamics in a coordination game between genes.  Importantly, the utility maximized in this game, as well as the amount by which each allele is boosted, is precisely the allele's {\em mixability}, or average fitness, a quantity recently proposed in \cite{mix} as a novel concept that is crucial in understanding natural selection under sex, thus providing a rigorous demonstration of that insight.  We also prove that the equilibria in two-person coordination  games can have large supports, and thus genetic diversity does not suffer much at equilibrium.  Establishing large supports involves answering through a novel technique the following question:  what is the probability that for a random square matrix $A$ both systems $Ax=1$ and $A^Ty=1$ have positive solutions?  Both the question and the technique may be of broader interest.

\end{abstract}

\newpage
\section{Introduction}
There was a crisis in Evolution in the early 20th century.  Mendel's insights and laws became known and largely accepted, except that their discrete nature was broadly perceived to contradict Darwin's theory, which was apparently about continuous traits and gradual change.   In the 1920s, this crisis was resolved, largely thanks to a mathematical formalism developed by Fisher, Wright, Haldane, and others.  Fitness was codified as the organism's expected offspring, and the equations of population genetics applied to genotypes demonstrated the possibility of continuous variation and change based on discrete genetic materials.   Ninety years later, this ``modern synthesis'' is the predominant rigorous way of understanding Evolution.

One of the important remaining open problems in Evolution today is the role of sex (often called ``the queen of problems in Evolution'').  Sex (creating offspring by recombining the genome of two organisms) is ubiquitous across life despite the fact that it is very costly:  it entails evolutionarily complex biological processes and behaviors that are risky and energy-inefficient, it breaks down lucky gene combinations, and it dilutes the parent's genetic contribution to the offspring, to mention a few.  As summarized in an important survey \cite{whysex}, all extant theories about the ``advantage of sex'' that counterbalances these costs fail to explain the mystery.  

In a recent paper \cite{mix}, a novel theory of the role of sex was proposed:  Natural selection under sex does not select allele combinations which optimize fitness (the default tacit assumption throughout research in Evolution), but instead promotes alleles that have high {\em mixability,} the ability to function adequately with a broad spectrum of genetic partners.   Under asex, of course, natural selection does select the fittest combinations.\footnote{To picture the difference, imagine a two-gene organism, where each gene comes in $n$ alleles (variants).  The fitness of the organism is captured now by an $n\times n$ matrix (for more genes, by a tensor) of positive numbers, sometimes called the {\em fitness landscape} of the species.  Without sex, natural selection will select the largest number in this matrix.  In contrast, sex will select {\em the rows and columns with the highest average} --- average, of course, with respect to the population's statistics.} In other words, the breakdown of lucky genetic combinations, which had been considered the major bug of sex (and which all theories of the advantage of sex had been attempting to explain away) is, in fact, sex's central {\em feature.}   

In the paper introducing mixability\cite{mix}, the central role played by mixability in natural selection under sex was demonstrated through extensive simulations on random fitness functions.  Any hope at a mathematical proof was hindered by the nonlinear nature of the equations.  The concept {\em was} proved for  certain cases of fitness functions \cite{jtb}, but these were highly specialized and symmetric cases for which the population genetic equations are linearized.  {\em In this paper we provide such a demonstration; in doing so, we make some totally unexpected connections between Evolution and familiar concepts from Computation and Game Theory.}

To achieve this, we explore the equations of population genetics in an important regime known as {\em weak selection} \cite{burger}.  Weak selection is a mathematical approximation of Kimura's {\em neutral theory}, a standpoint on Evolution that is presently broadly accepted (with modifications, of course) stating that most mutations do not affect  fitness (much).  Mathematically, in weak selection the elements of the fitness landscape are in $[1-s, 1+s]$ for some very small $s>0$.  In this regime, it is known \cite{nagy} that the state of the system (genotype frequencies) converge quickly to the so-called {\em Wright manifold}, where the distribution of genotypes is a product distribution of the allele frequencies in the population (see Section 2).   In other words, in this important regime of fitness functions, we can assume that the distribution of genotypes is a product distribution (the alleles of the various genes are chosen to form an offspring's genotype independently and with probability equal to their frequencies in the population), and thus we only need to track the allele frequencies  of the genes.  The distance of the genotype distribution from the product distribution (the distance from the Wright manifold) is called {\em linkage disequilibrium} and is of great interest in Evolution.  Under weak selection, it vanishes.

In this framework, we notice a remarkable convergence between Evolution and familiar concepts:

\begin{itemize}
\item   The frequency of an allele changes from one generation to the next by a proportional increment/decrement that is essentially {\em the incremental mixability of the allele.}

\item   The dynamics through which frequencies are updated are game dynamics of a {\em coordination game} \cite{jackson2002formation} (in a coordination game all players have the same utility for all strategy combinations).  The players are the genes, the strategies are the alleles, and the probability of playing a strategy is the frequency of the allele in the population.

\item The expected utility of an allele/strategy is precisely the mixability of this allele.

\item  The population genetics dynamics is thus precisely the {\em multiplicative update dynamics} of the game.  Notice that we are not saying that the equilibrium can be found through multiplicative updates; we are saying something much stronger: that Evolution {\em is} multiplicative updates.

\item  The process does converge to an equilibrium.  This was known  (e.g., \cite{nagy}), but now there is an easy proof through multiplicative updates, with mixability playing the role of regret.  In fact, this new proof establishes that convergence is robust under slight variants of the fitness values (Corollary 5)

\item  The equilibrium frequencies constitute a mixed Nash equilibrium of the subgame formed by their supports.

\end{itemize}

The fact that, at each step, the frequency of an allele is boosted by an amount equal to the mixability of the allele, provides the clearest possible verification of the hypothesis in \cite{mix} that an allele's mixability is what is being selected for in natural selection under sex.

Of course, there had been connections between game theory and evolution before, most notably Evolutionary Game Theory \cite{egt} whereby one models how the strategic behavior of a species is affected by the evolutionary advantages that it presents.  However, it is fair to say that many more game theorists have been interested in this connection than evolution researchers.  Our result connects game theory to the core of evolution theory, by providing an unusual point of view:  Evolution is a coordination game between genes, where the mixed strategy of each gene is the distribution of the alleles in the population, and the dynamics are fixed to multiplicative updates.  The genes are the players while the population stores the state.

\subsection*{Mixability and Diversity}
Why is mixability advantageous?  (Since sex is ubiquitous, there must be an advantage.)  One intuitive explanation may be that, by promoting ``good mixer genes'' (a phrase by Kimura \cite{kim} who had anticipated some of this thinking), it enhances genetic diversity.  But what happens at equilibrium (when the allele frequencies finally stabilize, as predicted by our results)?\footnote{In actual Evolution, of course, nothing ever stabilizes, as new mutations introduce new strategies in the game, and life goes on.  Note that our model does not include mutations.}   It would be disappointing --- and, by the above intuition, detrimental to the argument in favor of mixability --- if it so happens that, in the end, natural selection by sex ends up putting all its chips on a single allele per gene (as it often happened in the simulations in \cite{mix}).  The question arises:  {\em How large is the  typical support  of an equilibrium for our population genetics process?}  

In the second half of the paper (Section 4) we answer this question in a positive way:  there {\em are} exponentially many equilibria whose support contains a significant fraction of the alleles of each gene, see Corollary \ref{numbereq}.  We do this through a digression to a completely different problem:  Let $A$ be an $n\times n$ matrix.  What are the chances that the solution to the system $Ax=1$ is positive?   Assume that the entries of the matrix are iid from a distribution that is continuous and symmetric around zero, say uniform in $[-1,1]$; in this case, with probability one the solution exists and has no zero term.   Intuitively, the answer is $2^{-n}$ (each row of the inverse must have a positive sum).  It is not hard to show that this insight is correct.  

But suppose that we want both systems $Ax=1$ and $A^Ty =1$ to have positive solutions.  What are the chances of that happening?  One expects this to be about $2^{-2n}$, but there is dependence now and the calculation is not straightforward.  Intuitively the dependence is favorable, but how does one establish this?

We prove that this probability is {\em at least $2^{-2n+1}$} (and thus the dependence {\em is} indeed favorable, see Theorem \ref{diversitythm}).  The proof uses a potential function argument reminiscent of the Berlekamp switching game \cite{berl}.

The connection to Evolution is the following:  First, one can show that only square submatrices of the fitness matrix are likely to support an equilibrium of the population dynamics.  A square submatrix of the fitness matrix is the support an equilibrium if and only if the corresponding submatrix of {\em selections} (fitness minus one) has the property that the solutions to both row and column linear systems with unit right-hand sides are positive.  By showing that this happens with sufficiently high probability we establish that, in expectation at least, there are equilibria with substantial supports (in fact, quite a few of them), and thus diversity is not always lost at equilibrium.  

\section{Background}
A {\em gene} is a molecular unit of heredity.  Genes in organisms come in variants called {\em alleles}.  A combination of alleles, one for each gene, is called a {\em genotype.}  In this paper we study how the statistical make-up of a population of genotypes changes from generation to generation in a sexual species.  Members of the new generation are produced through {\em recombination,} which is the formation of a new genotype by choosing alleles from two parent genotypes in the previous generation. 

We make several (more or less standard) simplifying assumptions, which are generally trusted not to change substantially the essence of the evolutionary dynamics.  The population of genotypes is infinite.  We assume that the genotypes are {\em haploid} (even though humans are not) in that they contain only one copy of each gene, and that the organisms mate at random and {\em panmictically} (the latter meaning that the population is not bipartite to distinguish between males and females) to produce a new generation; further, we assume there is no overlap between generations (as if all mating happens simultaneously and soon before death). Each offspring's genome is formed by picking, for each gene, an allele from one of the two parent genomes, independently and with probability half each (that is, we assume that there is no {\em linkage}, i.e. genes that are close together on the same chromosome and thus are likely to be inherited together from the same parent).  

Our exposition will be for the case of two genes with $m$ and $n$ alleles respectively, even though our results can be easily seen to extend to any number of genes (with the exception of the results in Section 4 which are particular to matrices); we shall occasionally point out how the straightforward generalization proceeds.   Thus genotypes are pairs $ij$.  Each genotype $ij$ has a {\em fitness value} $w_{ij}$ which is the expected number of offspring the genotype produces (by mating randomly).  The matrix W, often called the {\em fitness landscape} of the species, entails the basic genetic parameters of the species (it is a $k$-dimensional tensor for $k$ genes).

We shall be interested in the statistics of the genotypes in the population.  The frequency of the genotype $ij$ will be denoted $p_{ij}$.  The matrix of the $p_{ij}$'s is the {\em state} of the dynamical system we shall follow.  We denote the value of $p_{ij}$ in generation $t$ by $p_{ij}^t$.

How do the $p_{ij}^t$'s change from one generation to the next?  Each member of the genotype $ij$ mates $w_{ij}$ times, each time with an organism with a random genotype, and an offspring genotype is created, whose alleles are chosen from the two parents at random and independently.  Then, it is easy to see that the expected frequency of genotype $ij$ at the next generation, $p_{ij}^{t+1}$, can be written:
$$p_{ij}^{t+1} = {1\over \bar w_{t}^{2}} \sum_{l} p_{il}^{t}w_{il} \sum_{k} p_{kl}^{t}w_{kj} $$
where $\bar w_{t}=\sum_{ij} p_{ij}^{t}w_{ij} $ is the average fitness of the population needed for normalization, so that frequencies add to one.

\subsection*{Wright Manifold, Weak Selection, and Nagylaki's Theorem}
Besides the $p_{ij}$ frequencies, one has the marginal frequencies, one for each allele:  $x_i=\sum_j p_{ij}$ and $y_j=\sum_ip_{ij}$.  Within the simplex of the $p_{ij}$'s, of particular interest to us is the {\em Wright manifold} on which $p_{ij}$ is a product distribution (the matrix $p_{ij}$ has rank one):  $p_{ij}=x_i\cdot y_j$.  It turns out that, on the Wright manifold, the population genetic equations take a much simpler form, expressed in terms of the marginal probabilities $x_i$ and $y_j$ (see Lemma 2).  

Life, in general, does not reside on the Wright manifold --- that is to say, genotype frequencies do not in general have rank one.  This is called {\em linkage disequilibrium}, and is measured by the distance from the Wright manifold $D_{ij} = p_{ij}-x_i\cdot y_j$.  Intuitively, it comes about because differences in the fitness of genotypes distort the allele statistics; just imagine two alleles of two genes whose combination is deleterious.  By definition, $D_{ij}$ is zero on the Wright manifold.

{\em Weak selection} is an important point of view on Evolution, which postulates that the entries of the tensor $W$ are all very close to one another.  Differences in fitness are minuscule, and the $w_{ij}$'s all lie within the interval $[1-s,1+s]$ for some very small $s>0$ which we call the {\em selection strength}.  It is a mathematical embodiment of the {\em neutral theory} of Kimura \cite{kimura1985neutral}, stating roughly that Evolution proceeds mostly through statistical drift due to sampling error that has no impact on fitness.  

There is an important connection between the Wright manifold and weak selection, best articulated through Nagylaki's Theorem.  Consider the evolution of genome frequencies $p_{ij}^t$ (or for more that two genes) in a situation in which the fitness values are within $[1-s,1+s]$ for some tiny $s>0$ --- that is, weak selection prevails.  Consider also the corresponding time series of linkage disequilibria $D_{ij}^t = p^t_{ij}-x_i\cdot y_j$.   

\begin{thm} {\bf (Nagylaki \cite{nagy,nagy2})}  (1) for any $t\geq t_0 = 3\log{1\over s}$ and any $i,j$, $D^t_{ij} = O(s)$; and furthermore

\smallskip\noindent (2) for $t\geq t_0$ there is a corresponding process $\{\hat p_{ij}\}$ on the Wright manifold such that (a) $|\hat p_{ij}^t - p_{ij}^t| = O(s)$; and (b) both processes converge and there is one-to-one correspondence between  the equilibria of $p_{ij}^t$ and the equilibria of $\hat p_{ij}^t$.  
\end{thm}

Nagylaki's Theorem states essentially that, in order to understand a genotype frequency process in the weak selection regime, one can instead follow a closely related process on the Wright manifold.  As we shall see next, it turns out that this simplifies things tremendously and brings about some unexpected connections.

\section{Coordination Games between Genes}
In a {\em game} each of finitely many {\em players} has a set of {\em strategies}, and a {\em payoff function} mapping the cartesian product of the strategy sets to the reals.  A game in which all payoff functions are identical is called a {\em coordination} game.  In a coordination game the interests of all players are perfectly aligned, and, intuitively, they all strive to hit the same high value of the common payoff function.  In terms of equilibrium calculations, they are trivial.

Fix a game, and a mixed strategy profile, that is, for each player $p$ a distribution $x^p$ over her strategies.  For each player $p$ and each strategy $a \in S_p$ one can calculate the expected payoff of this strategy, call it $q(a)$.  The {\em multiplicative update dynamics} of the game (see, for example, \cite{kalethesis}) transform the mixed strategy profile $\{x^p\}$ as follows:  For each player $p$ and each strategy $a\in S_p$, the probability $x^p(a)$ of player $p$ playing $a$ becomes 
$${x^p(a)\cdot (1+\epsilon\cdot q(a))\over 1+ \epsilon\cdot \sum_{b\in S_p}x^p(b)q(b)} = {x^p(a)\cdot (1+\epsilon\cdot q(a))\over 1+ \epsilon\cdot \bar q}, $$  
where by $\bar q$ we denote the expected payoff to $p$ (in a coordination game, to all players).  That is, the probability of playing $a$ is boosted by an amount proportional to its expected payoff, and then renormalized.  It is known that two players following the multiplicative update dynamics attain the Nash equilibrium in zero-sum games (this has been rediscovered many times over the past fifty years, see for example \cite{kalethesis}), but not in general games \cite{blum2007learning}.  Beyond games, and more importantly, the multiplicative updates dynamics lies at the foundations of a very simple, intuitive, robust and powerful algorithmic idea of very broad applicability \cite{cesa2006,kalethesis}.

Going now back to population genetics dynamics, let $W_{ij}$ be a fitness landscape (matrix for two genes, tensor for more) in the weak selection regime, that is, each entry is in the interval $[1-s,1+s]$.  Define the {\em differential fitness landscape} to be the tensor with entries $\Delta_{ij}={W_{ij}-1\over s}$.

We next point out a useful way to express the important analytical simplification afforded by the Wright manifold:

\begin{lem} On the Wright manifold, the population genetics dynamics becomes 
$$p_{ij}^{t+1} = {1\over \bar w_{t}} x_i^t\cdot y_j^t\cdot w_{ij},$$
and similarly for more genes.
\end{lem}
\begin{proof}
We have been using a special case of the general recombination dynamics, with $r=1$. For the general dynamics, as is shown in \cite{mix}, we can re-write the population genetics dynamics as:
$$p_{ij}^{t+1} = \frac{1}{\bar{w}_{t}} w_{ij}( p_{ij}^{t} - r D_{ij}^{t})$$
where $D_{ij}^{t} = p_{ij}^{t} - x_{i}^{t} y_{i}^{t}$ is the linkage disequilibrium. On the Wright Manifold, the linkage disequilibrium is equal to 0, and so we have $p_{ij}^{(t+1)} = \frac{1}{\bar{w_{t}}} p_{ij}^{(t)} w_{ij}$. Finally, since $D_{ij}^{t} = 0$, $p_{ij}^{t} = x_{i}^{t} y_{j}^{t}$. The result follows.  
\end{proof}

We are now ready for the main result of this section:

\begin{thm}
Under weak selection with selection strength $s$, the population genetic dynamics is precisely the multiplicative update dynamics of a coordination game whose payoff matrix is the differential fitness landscape and $\epsilon = s$.
\end{thm}

\begin{proof}
We only show the derivation for two genes, the general case being a straightforward generalization.  
$$x^{t+1}_i = \sum_j p_{ij}^{t+1}
 = {1\over \bar w_{t}}\sum_j x_i^t y_j^t w_{ij}
 = {x_i^t\over \bar w_{t}}(1+s\sum_j  y_j^t \Delta_{ij})
 =  {x_i^t\cdot (1+s\sum_j  y_j^t \Delta_{ij})\over 1 + s\cdot \bar \Delta }.
 $$
Here the first equation is the definition of marginal frequencies, the second is the Lemma, the third is the definition of $\Delta_{ij}$, and the last one recalls that the expectation of $w_{ij}$ is the one plus $s$ times the expectation of $\Delta_{ij}$.  The last expression is precisely the multiplicative update dynamics, completing the proof.
\end{proof}

It is known that the population genetics dynamics in the weak selection regime converge \cite{nagy,nagy2}; however, our result yields an easy proof. Moreover, we are able to show that the dynamics asymptotically achieve population frequencies that are comparable to the $i,j$ pair for which the cumulative mixability is maximized:

\begin{cor}
Let $m_{x}^{t}(i) = \sum_{j} y_{j}^{t}\Delta_{ij}$ and $m_{y}^{t}(j) = \sum_{i} x_{i}^{t}\Delta_{ij}$. The population genetics dynamics converge asymptotically in the weak selection regime to the genotype with maximal cumulative mixability. More formally, consider $i,j$ that maximizes 
$$\frac{1}{T} (\sum_{t=1}^{T} m_{x}(i)^{t} +m_{y}(j)^{t} + s \sum_{t=1}^{T} |m_{x}(i)^{t}| +|m_{y}(j)^{t}|)$$ 
it also holds that:  
$$\frac{1}{T} (\sum_{t=1}^{T} m_{x}^{t} \cdot x^{t}+ m_{y}^{t}\cdot y^{t} - \sum_{t=1}^{T} m_{x}(i)^{t} +m_{y}(j)^{t} + s \sum_{t=1}^{T} |m_{x}(i)^{t}| +|m_{y}(j)^{t}|) \geq 0$$ 
in the limit as $T \rightarrow \infty$ and thus that we do no worse than the best possible $i,j$ pair, and this $i,j$ pair also maximizes the cumulative mixability $\frac{1}{T} \sum_{t=1}^{T} 2+s(m_{x}(i)^{t} +m_{y}(j)^{t})$.  
\end{cor}
\begin{proof}
We have established in the Theorem that the population genetics dynamics is the multiplicative update dynamics with $\epsilon=s > 0$ for each player. We define some notation and then present the main tools needed for the convergence guarantee.\\
\\
It holds for the weak selection regime that $m_{x}^{t}$ and $m_{y}^{t}$ are in the range $[-1,1]$. Therefore by Theorem 3 in \cite{kalethesis} (and his remark that the same bound holds for our form of the multiplicative update dynamics) we can thus guarantee that, for any pair $i$ and $j$:
$$\sum_{t=1}^{T}  m_{x}^{t} \cdot x^{t}- m_{y}^{t}\cdot y^{t} - \sum_{t=1}^{T} m_{x}(i)^{t} +m_{y}(j)^{t} + s \sum_{t=1}^{T} |m_{x}(i)^{t}| +|m_{y}(j)^{t}| \geq -\frac{\log n+\log m}{s}$$  
clearly dividing both sides by $T$ and taking the limit as $T$ goes to infinity gives us 0 on the RHS. In particular, it will hold for the $i$ and $j$ that maximize $m_{x}$ and $m_{y}$. This shows that as $T \rightarrow \infty$, the dynamics converge to the $i$ and $j$ that maximize $\frac{1}{T}\sum_{t=1}^{T} 2+s(m_{x}(i)^{t} +m_{y}(j)^{t})$. The result follows.     
\end{proof}

In fact, we can further show that convergence is robust under perturbations of the fitness landscape:

\begin{cor}
Let $\tilde{w}_{ij}$ be subject to the weak selection regime, and thus be of the form $1+s\Delta_{ij}$. Let the differential fitness $w_{ij}^{t}$ be  $\tilde{w}_{ij} + \nu_{ij}^{t}$ with $\nu_{ij}^{t}$ drawn i.i.d from a zero-mean distribution and bounded almost surely in the interval $[-s(1-|\Delta_{ij}|),s(1-|\Delta_{ij}|)]$. If we run the population genetics dynamics with this fitness, then we will converge to the same equilibrium as when the differential fitness is $\tilde{w}$.  
\end{cor}
\begin{proof}
First we will show that the differential fitness at any point in time is subject to the weak selection regime. For any $t$, by the almost sureness of the bound on $\nu_{ij}^{t}$ it can be written in the form $1+s(\Delta_{ij}+ \nu_{ij}^{t}/s)$, where $|\nu|/s \leq 1-|\Delta_{ij}|$. Therefore the differential fitness still satisfies weak selection since $|\Delta_{ij}+ \nu/s| \leq |\Delta_{ij}|+|\nu_{ij}^{t}|/s \leq 1$. The convergence result from the Corollary therefore still holds for the dynamics in this case. The $i,j$ that is converged to maximizes:
$$ \frac{1}{T}\sum_{t=1}^{T} (2+s(\sum_{j} y_{j}^{t}\Delta_{ij} +\sum_{i} x_{i}^{t}\Delta_{ij}) + \nu_{ij}) = \frac{1}{T}\sum_{t=1}^{T} 2+s(\sum_{j} y_{j}^{t}\Delta_{ij} +\sum_{i} x_{i}^{t}\Delta_{ij}) + \frac{1}{T}\sum_{t=1}^{T} \nu_{ij}^{t}$$
  Because we are analyzing asymptotic behavior at equilibrium, we take the limit $T \rightarrow \infty$. In this case, for all $i,j$, $\frac{1}{T}\sum_{t=1}^{T} \nu_{ij}^{t} = 0$ by the strong law of large numbers since $E[\nu_{ij}^{t}]=0$. But then one is simply converging to $\frac{1}{T}\sum_{t=1}^{T} (2+s(\sum_{j} y_{j}^{t}\Delta_{ij} +\sum_{i} x_{i}^{t}\Delta_{ij}) + \nu_{ij})$, which is the same as what is being converged to for $\bar{w}$. Thus it holds for the $i,j$ that maximizes $\frac{1}{T}\sum_{t=1}^{T} (2+s(\sum_{j} y_{j}^{t}\Delta_{ij} +\sum_{i} x_{i}^{t}\Delta_{ij})$.   
\end{proof}

\section{Diversity}

We know that the population genetic dynamics converges, but what is the nature of the equilibria?  In particular, are they likely to have extensive support?  This would imply that diversity is not totally lost in the process, and considerably strengthen our model as a candidate for the role of sex in Evolution.  

To tackle this problem, we need a probabilistic model on the fitness landscape.  As we are assuming weak selection, we postulate that the entries of $W$ are drawn iid from a continuous distribution on $[1-s,1+s]$ with no singularities {\em that is symmetric around 1}.  Equivalently, the entries of $\Delta$ are iid on $[-1,1]$.   (In fact, our results do not require that the distributions of the entries be identical.)

\def\1{\hbox{\bf 1}}

At equilibrium, all alleles of a gene must have the same mixability (expected fitness with respect to the frequencies of the other alleles).  Focusing on the two gene case (here by necessity, because the more general case seems intractable), it must be that $Wx = a\1$ and $W^T y = b\1$ for some real vectors with nonnegative coefficients $x, y$ and some reals $a,b$ (in fact, it is easy to see that $a$ will be equal to $b$).  An equilibrium is characterized by the supports of $x$ and $y$, or, equivalently, by the submatrix defined by these.  Let us call a submatrix $A$ of $W$ an {\em equilibrium} if $Ax=a\1$ and $A^Ty=a\1$ have nonnegative solutions $x,y$ adding to one, for some $a>1$.  We require that $a>1$ for this reason:  If $a<1$, then $A$ is indeed an equilibrium, but one that leads to extinction.

Consider an $m\times n$ submatrix $A$ of the fitness matrix.  When is $A$ an equilibrium?  First of all, if $A$ is not square, say $m<n$, then the probability of $A$ being an equilibrium is zero, because then the system $A^Ty = a\1$ is overdetermined.  So, we shall focus on a square submatrix $A$.  Under weak selection, we can write $A= U + sB$.

\begin{lem}
$A$ is an equilibrium if and only if $B^{-1}$ exists and has positive row and column sums.
\end{lem}

\begin{proof} 
If $A$ is an equilibrium with $x,y>0$ the solutions to the linear systems with right-hand sides $a\1$, 
then it is easy to see that $Bx = {\1}(a-\sum_j y_j)=\1(a-1)$ and similarly $B^T y = {\1}(a-1)$.  Therefore $B^{-1}$ must exist and have positive row and column sums.  And from any nonnegative solutions of $Bz=\1,B^Tw = \1$ we can get back the solutions of $Ax=a\1,A^ty=a\1,\sum_i x_i=1,\sum_j y_j=1$, adding up to one:  $x={z\over \sum_i z_i}$ and similarly $y={w\over \sum_j w_j})$ with $a=1+{s\over \sum_i z_i}>1$.  
\end{proof}

Thus, to show that there are enough equilibria with large support, we must calculate (more precisely, lower bound) the probability that a random matrix $B$ has an inverse with positive row and column sums.  Our main result is the following:

\begin{thm}  \label{diversitythm}
The probability that $A$ is an equilibrium is at least $2^{-(2n-1)}$.
\end{thm}

\begin{proof}
By the Lemma, we must bound from below the probability that $B^{-1}$ has positive sums (since $B^{-1}$ exists with probability one).

Let $S\subseteq [n]$.  By $I_S$ we denote the $n \times n$ unit matrix with the $i$th one replaced by $-1$ whenever $i\in S$.  Notice that $I_S E I_T$ is E with all columns in $S$ flipped (multiplied by $-1$) and all rows in $T$ also flipped (with entries at the intersection of a row in $T$ and a column in $S$ restored to its original value).  Notice that $(I_SEI_T)^{-1} = I_TE^{-1}I_S$.   That is, to invert $E$ with some rows and columns flipped, you invert $E$ and then flip the same columns and rows, with the roles of columns and rows interchanged.

Now take $B$ and consider all possible flippings $I_SBI_T$.  There are $2^{2n-1}$ distinct such matrices, because it is easy to see that $I_SBI_T= I_{[n]-S}BI_{[n]-T}$ and that all other pairs of flippings are distinct.  We shall argue that one of these flippings must have positive row and column sums.

\begin{lem}
For every $B\in[-1,1]^{n\times n}$ there are $S,T\subseteq [n]$ such that $I_SBI_T$ has nonnegative row and column sums.
\end{lem}

\begin{proof} To prove the lemma, start with $B$ and perform the following:

\medskip\noindent\centerline{\em while there is a row or column with negative sum, flip it.}

\medskip Naturally, after each such flipping other columns or rows, which had positive sums, may become negative.  However, if the sum of the row or column that was flipped was $-\sigma$, notice two things:  First, $\sigma\geq \sigma_0$, where $\sigma_0$ is a constant depending on $B$; and second, {\em the total sum of the entries of $E$ increases by $2\sigma$ after the flip}.   Therefore, the process must end, and this can only happen if the matrix has positive row and column sums.
\end{proof} 

The theorem now follows:  Consider the domain $M=[-1,+1]^{n\times n}$, and the subset $M_+$ whose inverse exists and has positive sums.  This subset can be defined using polynomial inequalities, and is this measurable.  By applying the $2^{2n-1}$ transformations $B\mapsto I_S B I_T$ to $M_+$, by the Lemma we exhaust all of  $M$.  Therefore, the probability that a matrix in $M$ is in $M_+$ is at least $2^{-(2n-1)}$.
\end{proof}

\begin{cor} \label{numbereq}
The expected number of $k\times k$ equilibria in an $m\times n$ weak selection fitness matrix  is at least   $2\left({mn\over 4k^2}\right) ^k$.
\end{cor}

Notice that, for $k$ less than half times the geometric mean of $m$ and $n$, this is exponential in $k$.

\medskip 

\begin{proof}
Use the approximation  ${n \choose k} \geq ({n\over k})^k$.
\end{proof}

\section{Discussion and Open Problems}

Les Valiant \cite{valiant} was the first to point out a connection between Evolution and Learning:  The class of traits which are achievable by a species through random mutations constitute a specialized kind of learnability (a subcase of statistical learning, actually).  Here we point out a different connection between these two fundamental computational categories:  On the Wright manifold (that is, when linkage disequilibrium can be disregarded), Evolution {\em is} in fact learning through multiplicative updates --- a well-known learning algorithm of very broad appeal. The same reasoning establishes Evolution as a coordination game between genes.  Genes are the players, alleles are the strategies, and the mixed strategies are the allele statistics in the population.  Notice how the organism is sidelined in this viewpoint:  Evolution is an interaction between the genes (acting in a seemingly strategic manner) and the population (which stores the system's state). This is consistent with the gene-centric view of evolution, by which much can be explained by just focusing on evolution of and interaction between genes rather than organisms or species \cite{szathmary1995major}.

There is an important difference between Valiant's work and our point of view:  Valiant's evolvability is concerned with understanding random mutations.  We are concerned in the ways in which sex and recombination nudge allele frequencies in the direction of high mixability (expected fitness); mutations are only implicit in our model, they are the (much slower) background process that creates the diversity exploited by recombination. Evolution and coordination games are not to be thought of as identical, however. A key difference between Evolution and coordination games is this:  In a coordination game, a player may switch to a strategy which is currently not in the player's support (is being played with probability zero), once it becomes favorable.  In Evolution, once an allele becomes extinct it never comes back (except through a new mutation, at a time scale far bigger than our current concerns).  That is why the equilibrium of the process may not be a Nash equilibrium of the original game (but it is a mixed Nash equilibrium of the subgame defined by the support).

Starting with Darwin, researchers have often expressed an instinctive sense of disbelief that all Life we see around us could have come about through the rather rudimentary processes envisioned by Evolution.  Connecting Evolution with multiplicative updates may help a little in this cognitive/cultural difficulty; after all, the multiplicative updates algorithm has surprised us time and again with its seemingly miraculous performance and applicability.

This work leaves open a variety of questions. Besides weak selection, for which other classes of fitness landscapes is the present analysis applicable?  It is known that {\em product landscapes} (the fitness of a genotype is the product of fitness values, one for each allele present in the genotype) have the property of {\em staying} on the Wright manifold, once they are started there, but are not in general attracted to it \cite{burger}.  It would be very interesting to come up with a combinatorial characterization of the fitness landscapes that converge to the Wright manifold (like the weak selection landscapes) and of those which at least stay there (like the multiplicative landscapes).

Secondly, in the context of the diversity proof in the previous section, consider the following computational problem: ``flip the rows and columns of a given nonsingular square matrix such that the inverse  has nonnegative row and column sums.''  The existence proof through a potential function places this problem in the class PLS \cite{jpy}.  Is it PLS-complete?

Finally, we would love to extend our diversity result (Corollary \ref{numbereq})  to three or more genes.  Unfortunately, in this more general case the equations become multilinear, and such equations are very hard both to solve and to argue about.  Our simulations, however, show that large equilibria exist in the multi-gene setting as well.

\bibliographystyle{unsrt}
 \bibliography{evogames}

\end{document}